\documentclass{acm_proc_article-sp}
\pdfoutput=1
    
\usepackage{graphicx}

\usepackage{amsmath, amssymb, amsthm, thm-restate}

\usepackage[ruled]{algorithm2e}

\usepackage{myDefs}  

\newtheorem{definition}{Definition}

\newtheorem{lemma}{Lemma}
\newtheorem{theorem}{Theorem}

\title{Universally Optimal Privacy Mechanisms for Minimax Agents}

\makeatletter
\author{
  Mangesh Gupte\\
  Department of Computer Science,\\
  Rutgers University, Piscataway, NJ 08854.  \\
  mangesh@cs.rutgers.edu \\
  \and
  Mukund Sundararajan \\
  Google -- 1600 Amphitheatre Parkway, \\
  Mountain View,  CA 94043. \\
  mukunds@google.com 
}
\makeatother
\begin{document}
\maketitle

\begin{abstract}

  A scheme that publishes aggregate information about sensitive
  data must resolve the trade-off between utility to information
  consumers and privacy of the database participants. Differential
  privacy~\cite{BigBang} is a well-established definition of
  privacy---this is a universal guarantee against all attackers,
  whatever their side-information or intent. In this paper, we
  present a universal treatment of utility based on the standard
  minimax rule from decision theory~\cite{minimax} (in contrast to
  the utility model in~\cite{Bayesian}, which is Bayesian).
  
  In our model, information consumers are minimax (risk-averse)
  agents, each possessing some side-information about the query,
  and each endowed with a loss-function which models their
  tolerance to inaccuracies. Further, information consumers are
  rational in the sense that they actively combine information from
  the mechanism with their side-information in a way that minimizes
  their loss. Under this assumption of rational behavior, we show
  that for every fixed count query, a certain geometric mechanism
  is universally optimal for all minimax information
  consumers. Additionally, our solution makes it possible to
  release query results at multiple levels of privacy in a
  collusion-resistant manner.
\end{abstract}

\section{Introduction}

\textbf{Privacy Mechanisms:} Agencies such as medical establishments,
survey agencies, governments use and publish aggregate statistics
about individuals; this can have privacy implications. Consider the
query: \emph{Q: How many adults from San Diego contracted the flu this
  October?} The government can use the query result to track the
spread of flu, and drug companies can use it to plan production of
vaccines. However, knowledge that a \emph{specific} person contracted
the flu could be used to deny her health insurance based on the
rationale that she is susceptible to disease.  As discussed
in~\cite{Sweeney}, and as is exemplified by~\cite{AOL,Netflix},
seemingly benign data publications can have privacy
implications. Thus, it is important to think rigorously about
privacy. The framework of differential privacy~\cite{Survey} does
this, and is applicable widely (see Section~\ref{sec:related}).

Mechanisms guarantee differential privacy by perturbing results --
they add random noise to the query result, and guarantee protection
against all attackers, whatever their side-information or intent
(see Section~\ref{sec:basic} for a formal definition).

\textbf{Our Utility Model:} The addition of noise increases privacy
but intuitively reduces utility of the query result. To understand
this privacy-utility trade-off, we propose a formal
decision-theoretic model of utility. Decision-theory is a widely
applied field that provides mathematical foundations for dealing
with preferences under uncertainty. The use of decision theory in
this context is appropriate because, as we discussed above,
mechanisms guarantee differential privacy by introducing
uncertainty.

In our model of utility (see Section~\ref{sec:utility} for
details), the user of information, i.e. the \emph{information
  consumer} has \emph{side-information}---for instance, knowledge
of the population of San Diego is an upper bound on the result of
the query $Q$. It has a \emph{loss-function} that expresses it's
tolerance to inaccuracy. It is \emph{rational} in the sense that it
combines information from the mechanism with its side-information
optimally with respect to its personal loss-function. It is
\emph{risk-averse} in the sense that it would like to minimize
worst-case loss over all scenarios.\footnote{Ghosh et
  al.~\cite{Bayesian} propose a model with most of these features,
  but assumes that information consumers are Bayesian and have a
  prior over the query-result. }

Given the privacy parameter, the loss-function and the
side-information of an information consumer it is possible to
identify an optimal mechanism -- a mechanism that is differentially
private and that maximizes its utility. See Section ~\ref{sec:LP}
for an algorithm to find such a mechanism.

\textbf{Non-Interactive Settings:} Very often aggregate statistics,
like answers to Q, are \emph{published} in mass media as opposed to
following a query-response form~\cite{h1n1}. In such cases neither
the information consumer nor it's loss-function and
side-information are known in advance. Thus it seems hard to
identify the optimal mechanism for a information consumer.

Nevertheless, we show that it is possible to deploy an optimal
mechanism \emph{without knowledge of the information consumer's
  parameters}. Furthermore, this mechanism is \emph{universally
  optimal for all information consumers}, no matter what their
side-information or loss-function.

How can we identify the optimal mechanism without knowledge of the
information consumer's parameters? The apparent paradox is resolved
by relying on the information consumers' rationality, i.e., each
information consumer \emph{uses} its personal loss-function and
side-information to actively transform the output of the deployed
mechanism. For a certain class of queries called count queries,
when the deployed mechanism is a certain geometric mechanism, this
transformation is effective enough to result in the optimal
mechanism for the information consumer---a fact that we will
establish via linear-algebraic proof techniques.
  
\textbf{Multiple Levels of Privacy:} We also show how to
simultaneously release the query result at different levels of
privacy to different information consumers. This is useful, for
instance, when we want to construct two versions of the report on
flu statistics, one which prioritizes utility for the eyes of
government executives, and a publicly available Internet version
that prioritizes privacy.

A naive solution is to perturb the query results differently,
independently adding differing amounts of noise each time. The
drawback is that consumers at different levels of privacy can
collude and combine their results to \emph{cancel} the noise (as in
Chernoff bounds). An alternate way is to correlate the noise added
to different outputs. We give an algorithm to achieve this that
makes the data release collusion-resistant.

In this paper we focus on a single query; the complexity comes from
a rich model of consumer preferences, where we consider different
utility functions for each consumer and optimize for each of
them. \cite{Learning, Boosting, Geometry} exploit similarities
between the queries to obtain extension to multiple queries with
good utility guarantees. However, they do not consider a rich
consumer preference model. Our results could be used as a building
block while answering multiple queries.

\section{Model and Results}
\label{sec:model}
We gave a informal description of our model and results in the
Introduction. In this section, we formally define our model and
discuss the main results. The proofs of the results are presented
in Sections~\ref{sec:char},~\ref{sec:applications}.

\subsection{Privacy Mechanisms and Differential Privacy}
\label{sec:basic}

A \emph{database} is a collection of rows, one per individual. Each
row is drawn from an arbitrary domain $D$; for instance, in our
running example, a row of the database has the name, age, address
and medical records of a single individual. A database with $n$
rows is thus drawn from the domain $D^n$.

We will focus on a class of queries, called \emph{count queries},
that frequently occur in surveys: Given a predicate $p : D \to
\{\text{True, False}\}$, the result of a count query is the number
of rows that satisfy this predicate, a number between $0$ and the
database size, $n$. $Q$ is an example of a count query with the
predicate: \emph{individual is an adult residing in in San Diego,
  who contracted flu this October}. Though simple in form, count
queries are expressive because varying the predicate naturally
yields a rich space of queries.

We guarantee \emph{differential privacy} to protect information of
individual database participants. Differential privacy is a
standard, well-accepted definition of privacy~\cite{Survey} that
has been applied to query privacy~\cite{Datamining,BigBang,Smooth},
privacy preserving machine learning~\cite{Learning, Kavi08} and
economic mechanism design~\cite{MD}. A \emph{fixed} count query
maps the database $d$ to a number which belongs to the set $N$. A
privacy mechanism $M$ for a \emph{fixed} count query is a
probabilistic function that maps a database $d \in D^n$ to the
elements of the set $N = \{0 \ldots n\}$. These can be represented,
for each $d \in D^n$, by $\{m_{d,r} \}_{r \in N}$, which gives for
each database $d\in D^n$ the probability that $M$ outputs $r$. For
the database $d$, the mechanism releases a perturbed result by
sampling from the distribution $\{m_{d,r}\}_{r \in N}$.

The Geometric Mechanism~\cite{Bayesian} is a simple example of a
privacy mechanism. It is a discrete version of the Laplace
Mechanism from~\cite{BigBang}.

\begin{definition}[$\alpha$-Geometric Mechanism]
\label{def:geom}
  When the true query result is $f(d)$, the mechanism outputs $f(d)
  + Z$. $Z$ is a random variable distributed as a two-sided
  geometric distribution: $Pr[Z=z]=\frac{1-\alpha}{1+\alpha}
  \alpha^{|z|}$ for every integer~$z$. 
\end{definition}

\begin{figure}[ht]  
  \includegraphics[width=0.33\textwidth,angle=-90]{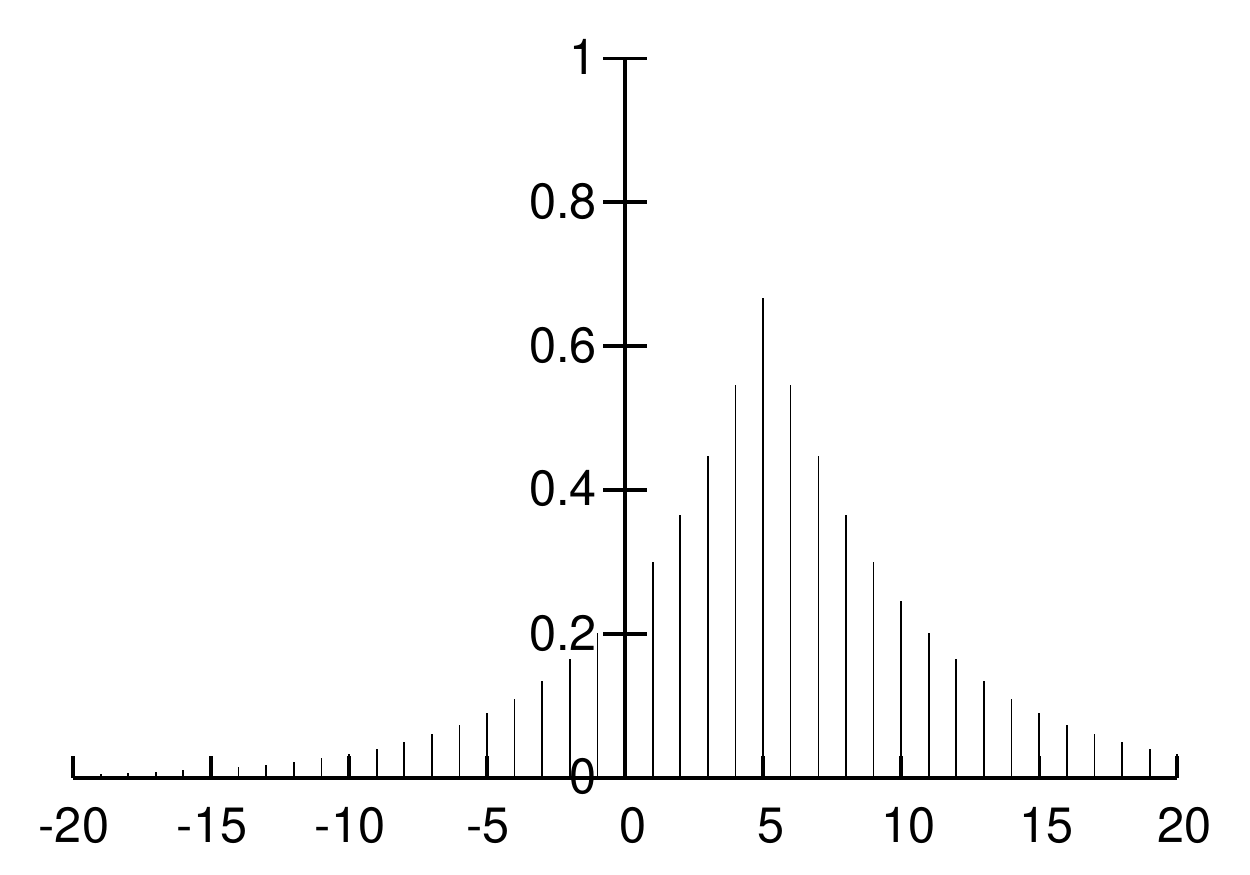}
  \caption{The probability distribution on outputs given by the
    Geometric Mechanism for $\alpha = 0.2$ and query result 5.}

\end{figure}

Informally, a mechanism satisfies differential privacy if it
induces similar output distributions for every two databases that
differ only in a single individual's data, thereby ensuring that
the output is not sensitive to any one individual's
data\footnote{Thus any attack on an individual's privacy that can
  be constructed using the perturbed query result with this
  individual present in the database can also be constructed, with
  a similar success rate, without this individual present in the
  database. See~See~\cite{BigBang, Semantic} for details of such
  semantics of differential privacy.}. Formally, differential
privacy is defined as follows~\cite{BigBang}:

Given a privacy parameter $\alpha \in [0,1]$ and two database $d_1,
d_2 \in D^n$ that differ in at most one individual's data, a
mechanism $M$ is $\alpha$-differentially private, if for all
elements $r$ in the range of the mechanism : $\reci{\alpha}\cdot
x_{d_1,r} \geq x_{d_2,r} \geq \alpha \cdot x_{d_1,r}$.

The parameter $\alpha$ can be varied in the interval $[0,1]$ to
vary the strength of the privacy guarantee---when $\alpha=0$, the
above definition is vacuous and there is no privacy, whereas when
$\alpha=1$, we effectively insist on absolute privacy-- the query
result cannot depend on the database because we require
distributions over perturbed results to be identical for
neighboring databases.

\subsection{Oblivious Mechanisms} 
 
We will focus in this paper on a class of privacy mechanisms that
are oblivious. A mechanism is \emph{oblivious} if it sets up an
identical distribution over outputs for every two databases that
have the same unperturbed query result. Naturally, an
implementation of an oblivious mechanism need only have access to
the true query result---the \emph{input}---and can be oblivious to
the database itself. An oblivious mechanism for count queries can
be expressed by the set of probability masses for every $i\in N$ :
$\{x_{i,r}\}_{r \in N}$, where $x_{i,r}$ is the probability that
the mechanism outputs $r$ when the true result is
$i$. Appendix~\ref{sec:nonoblivious} shows that this restriction to
oblivious mechanisms is without loss of generality. The geometric
mechanism (Definition~\ref{def:geom}) only depends on the query
result $f(d)$ and not on the database $d$ itself; so it is a
oblivious mechanism.

The query result for a count query can change by at most one when
we change any one row of the database, so we can rewrite the
definition of differential privacy as follows:

\begin{definition} [Differential Privacy for Count Queries]
  \label{def:privacy}
  An oblivious mechanism for count queries for $\alpha \in [0,1]$
  is $\alpha$-differentially private if for all $ i\in\{0 \ldots
  n-1\}, r\in N$ : \[\reci{\alpha} x_{i,r} \geq x_{i+1,r} \geq
  \alpha x_{i,r}. \]
\end{definition}

Observe that the geometric mechanism is $\alpha$-differentially
private because for two adjacent inputs $i,i+1 \in N$, and any
output $r \in N, \frac{x_{i,r}}{x_{i+1,r}} \in [\alpha, 1/\alpha]$.

\subsection{Minimax Information Consumers}
\label{sec:utility}

We now discuss our model of an information consumer's utility. The
loss-function $l(i,r): N\times N \to \real$ specifies the loss of the
information consumer, given the mechanism outputs $r$ when the true
result is $i$. We only assume that the loss-function is monotone
non-decreasing in $|i-r|$, for every $i$. That is, the consumer
becomes unhappier as the returned answer is further away from the
true result. 

Consider some examples of valid loss-functions: The loss-function
$l(i,r) = |i-r|$ quantifies the mean error---for our query $Q$,
this loss-function may be a reasonable one for the government who
want to keep track of the rise of flu. The loss-function $l(i,r) =
(i-r)^2$ quantifies the variance in the error---this may be
reasonable for a drug company who wants to ensure that they don't
over-produce or under-produce the flu drug. The loss-function
$l(i,r) = \begin{cases} 0 & \text{ if } i = r \\ 1 & \text{ if } i
  \neq r\end{cases}$, measures the frequency of error.

Additionally, we will assume that the information consumer has side
information $S \subseteq N$, i.e., the information consumer knows
that the query result cannot fall outside the set $S$. For
instance, knowledge of the population of San Diego yields an
upper-bound on the query result. The drug company may also know how
many people bought its flu drug this month, yielding a lower bound
on the query result.

For any \emph{specific} input $i$, the loss-function $l$ allows us
to evaluate the information consumer's dis-utility as the expected
loss over the coin tosses of the mechanism: $\sum_{r \in N} l(i,r)
\cdot x_{i,r}$.  To quantify the overall loss, we follow the
minimax decision rule, i.e., we take the worst-case loss over all
inputs in the set $S$~\cite{minimax}. This amounts to the
information consumers being risk-averse. Hence, the dis-utility of
the mechanism $x$ to the consumer $c$ is :
\begin{equation}
  L(x) = \max_{i\in S} \sum_{r \in N} l(i,r) \cdot x_{i,r} 
\end{equation}

\subsection{Interactions of Information Consumers with Mechanisms}

As mentioned in the Introduction, information consumers actively
interact with the mechanism to induce a new mechanism; we now discuss
the mechanics of this interaction.

\subsubsection{Motivation}

The following example argues why a rational information consumer
will not accept the mechanism's output at face value.

\begin{example} \label{ex:motive} Recall the query $Q$ defined in
  the Introduction. Suppose that the information consumer is a drug
  company, who knows that $l$ individuals in San Diego bought its
  flu drug in the month of October. Thus the query result $Q$ must
  be at least $l$; the information consumer cannot conclude that
  the query result is exactly $l$ because some individuals with flu
  may have bought a competitors drug, or bought no drug at
  all. Thus it has side-information $S = \{l \ldots n\}$.

  Suppose we deploy the geometric mechanism for the query $Q$. This
  mechanism returns with non-zero probability outputs outside the
  set $\{l \ldots n\}$. Such outputs are evidently incorrect to the
  information consumer, and naturally it makes sense for the
  information consumer to map these results within the set $\{l
  \ldots n\}$. Though it is not clear what the best way of doing so
  is, a reasonable rule may be to re-interpret results less than
  $l$ as $l$, and results larger than $n$ as $n$.
\end{example}

\subsubsection{Feasible Interactions}
\label{sec:interaction}
Before we discuss the optimal way for an information consumer to
interact with the mechanism, we describe the space of feasible
interactions. On receiving a query result $r$ from the mechanism,
the consumer can reinterpret it as a different output. This
reinterpretation can be probabilistic and can be represented by a
set of probability masses $\{ T_{r,r'} : r' \in N \}$ which gives
for each result $r$, the probability that the consumer will
reinterpret it as the output $r'$. Such an interaction induces a
new mechanism for the user. Suppose the deployed mechanism is
represented by the set of probability masses $\{y_{i,r} : i,r \in
N\}$, and the induced mechanism as the probability masses $
\{x_{i,r'} : i,r \in N \}$, then $x_{i,r'} = \sum_{r \in N} y_{i,r}
\cdot T_{r,r'}$. We formalize this in a definition.

\begin{definition}[Derivability]
  \label{def:der}
  Given two mechanism $x$ and $y$, we say that mechanism $x$ can be
  \emph{derived} from $y$ if and only if, for every $r\in N$, there
  exists a set of probability masses $\{ T_{r,r'} : r' \in N \}$
  such that for every $i,r'\in N $ : $x_{i,r'} = \sum_{r \in N}
  y_{i,r} \cdot T_{r,r'}$.
\end{definition}

\subsubsection{Optimal Interactions}
\label{sec:optinter}
\label{sec:LP}

Given a deployed mechanism $y$, the optimal interaction $T^*$ is
one that minimizes the information consumer's maximum loss on the
induced mechanism. The optimal interaction can be computed using a
simple linear program. There are $n^2$ variables: one for each
$T^*_{r,r'\in N}$.  The objective function is obtained by
minimizing the loss to the consumer if it uses interaction
$T^*$. The constraints are obtained from the fact for each $r$, the
entries $T^*_{r,r'}$ form a probability distribution and hence sum
up to 1 and that all entries of $T^*$ are positive. The actual
linear program is given as :
\begin{align*}
  \text{minimize} & \qquad \max_{i\in S} \sum_{r \in N} x_{i,r} \cdot l(i,r) \\
  x_{i,r} = \sum_{r \in N} y_{i,r} \cdot T^*_{r,r'} &\qquad \forall i \in N,\forall r \in N \\
  \sum_{r' \in N} T^*_{r,r'} = 1 &\qquad \forall r \in N \\
  T^*_{r,r'}\geq 0 &\qquad \forall r \in N, \forall r' \in N
\end{align*}

\subsection{Optimal Mechanism for a Single Known Information
  Consumer}
\label{sec:opt}

Identifying the optimal mechanism for a specific consumer reduces
to the following: Identify a level of privacy $\alpha$ with which
to release the result. Find the consumer's loss-function and
side-information. Identify an $\alpha$-differentially private
mechanism such that the mechanism induced by the consumer's optimal
interaction (as described in the previous section), has the best
possible utility.

In the case of a single information consumer, we can obviate the need
for the information consumer to reinterpret the deployed mechanism's
output: Suppose there is a mechanism $y$ with post-processing $T$ that
induces a mechanism $x$. Clearly, presenting $x$ directly to the
information consumer yields at least as much utility for it. All we
have to ensure is that $x$ is $\alpha$-differentially private, and a
simple proof (omitted) shows that this is indeed so.

Thus, to identify the optimal mechanism for a specific information
user, it suffices to search over $\alpha$-differential
mechanisms. 
For a given consumer $c$ with loss-function \[ L(l,S) = \max_{i\in
  S} \sum_{r \in N} x_{i,r}\cdot l(i,r) \] and privacy parameter
$\alpha$, the optimal differentially private mechanism $M_c$ is the
solution to a simple linear program. Like in the previous section,
there are $n^2$ variables one for each matrix entry of the
mechanism $x$. The objective is to minimize the user's loss
function. The constraints are obtained by the facts that
\begin{enumerate}
\item $x$ is differentially private. So the variables $x_{i,r}$
  must satisfy Definition~\ref{def:privacy}
\item For each input $i$, elements $x_{i,r}$ form a probability
  distribution and hence sum up to 1.
\item All elements $x_{i,r}$ are positive
\end{enumerate}
Writing this as an optimization problem we get: 

\begin{align*}
\text{minimize} & \qquad \max_{i\in S} \sum_{r \in N} x_{i,r} \cdot l(i,r) \\
x_{i,r} - \alpha \cdot x_{i+1,r} \geq 0 &\qquad \forall i \in N \setminus
\{n\}, \forall r\in N \\
\alpha \cdot x_{i,r} -  x_{i+1,r} \leq 0  &\qquad \forall i \in N \setminus
\{n\}, \forall r\in N \\
\sum_{r \in N} x_{i,r} = 1 &\qquad \forall i \in N \\
x_{i,r}\geq 0 &\qquad \forall i \in N, \forall r \in N
\end{align*}

We can convert it into a Linear Program, the solution of which gives
us $x^*$.

\begin{align*}
\text{minimize} & \qquad d\\
d - \sum_{r \in N} x_{i,r} \cdot l(i,r) \geq 0 &\qquad \forall i\in S\\
x_{i,r} - \alpha \cdot x_{i+1,r} \geq 0 &\qquad \forall i \in N \setminus
\{n\}, \forall r\in N \\
\alpha \cdot x_{i,r} -  x_{i+1,r} \leq 0  &\qquad \forall i \in N \setminus
\{n\}, \forall r\in N \\
\sum_{r \in N} x_{i,r} = 1 &\qquad \forall i \in N \\
x_{i,r}\geq 0 &\qquad \forall i \in N, \forall r \in N
\end{align*}

To deploy this mechanism $x^*$, we first compute the true query
result, say $i$, then sample the perturbed result $r$ from the
distribution $\{x^*_{i,r} : \forall r \in N\}$, and release the
sampled result. Table~\ref{tab:optmech}(a) gives an example of a
optimal mechanism for a particular information consumer.

\begin{table*}[ht]
  \begin{tabular}{ccc}
    {
      $
      \begin{bmatrix}
        2/3 & 5/17 & 1/25 & 1/98  \\
        1/6 & 7/11 & 7/44 & 2/49  \\
        2/49 & 7/44 & 7/11 &  1/6  \\
        1/98 & 1/25 & 5/17 &  2/3  \\
      \end{bmatrix}
      $
    }    
    &
    {
      $ 
      \begin{bmatrix}
        4/3 &   1/4 &  1/16 &  1/48 \\
        1/3 &     1 &   1/4 &  1/12 \\
        1/12 &   1/4 &     1 &   1/3 \\
        1/48 &  1/16 &   1/4 &   4/3 \\
      \end{bmatrix}
      $        
    }
    &
    {
      $ 
      \begin{bmatrix}
        9/11 & 2/11 &    0 &    0 \\
        0 &    1 &    0 &    0 \\
        0 &    0 &    1 &    0 \\
        0 &    0 & 2/11 & 9/11 \\
      \end{bmatrix}
      $
    }
    \\ 
    {(a)The Optimal Mechanism} & { (b)$G_{3,\frac{1}{4}}$ } &    { (c)Consumer
      Interaction}\\
    & Mechanism with access to the database. &  Mechanism with access to the user parameters.
 \end{tabular}
 \caption{This shows the optimal mechanism for a consumer $c$ with loss-function $l(i,r) =
   |i-r|$ and side-information $S = \{0,1,2,3 \}$. $n = 3, \alpha = 1/4$.
 }
  \label{tab:optmech}
\end{table*}

\subsection{Optimal Mechanism for Multiple Unknown Information
  Consumers}

How can we extend the results of the previous section to multiple
consumers? The naive solution is to identify and separately deploy
the optimal mechanism for each information consumer as described in
the previous section.

There are two reasons why this is undesirable. First, the naive
solution results in the release of several re-randomizations of the
query result---this allows colluding consumers to combine their
results and cancel out the noise leading to a degradation in
privacy; see~\cite{MD} for a discussion,

Second, solving the linear program that identifies the optimal
mechanism for a user requires the knowledge of the consumer's
parameters; knowledge that is often unavailable when the decision
of which mechanism to deploy is made. Consider a report published
on the Internet. It is not clear who the information consumers are
going to be.

Our main result works around these issues successfully.

\begin{restatable}{theorem}{ummthm}
  \label{thm:umm}
  Consider a database $d$, count query $q$, $k$ consumers and
  privacy levels $\alpha_1 < \ldots < \alpha_k$.  There exists a
  mechanism $M$ that constructs $k$ results $r_1\ldots r_k$, and
  releases result $r_i$ to the $i$th information consumer, such
  that:
  
  \begin{enumerate}
  \item(Collusion-Resistance) Mechanism $M$ is $\alpha_{i'}$-differentially private
    for any set $I$ of colluding information consumers who combine
    their results. Here, $C \subseteq \{1 \ldots k\}$ and $i'= \min\{j
    : j \in C\} $.
    
  \item(Simultaneous Utility Maximization) Suppose that the $i$th
    consumer is rational and interacts optimally with the mechanism
    (as described in Section~\ref{sec:optinter}), then its utility
    is equal to that of the differentially private mechanism
    tailored specifically for it (the mechanism from
    Section~\ref{sec:opt}).
  \end{enumerate}
\end{restatable}



We now describe the release mechanism $M$. The $i$th stage of the
mechanism $M_i$ is just the $\alpha_i$-geometric mechanism. We
shall prove in Lemma~\ref{lem:addpriv}, that for any $\alpha >
\beta$, the $\alpha$-geometric mechanism can be derived from the
$\beta$-geometric mechanism: that is there is an implementable
mechanism $T_{\alpha,\beta}$ such that if we use $T_{\alpha,\beta}$
to reinterpret results given by the $\beta$-geometric mechanism, we
get the $\alpha$-geometric mechanism. The query results $r_i$ are
not generated independently of each other, they are obtained by
successive perturbations: the result $r_i$ of mechanism $M_i$ is
given as input to the mechanism $T_i =
T_{\alpha_i,\alpha_{i+1}}$. Hence, the $(i+1)$th stage mechanism
$M_{i+1}$ is just the $\alpha_{i+1}$-geometric mechanism. This
specifies how the noise added to the query results is corelated. We
describe the mechanism formally in
Algorithm~\ref{alg:multilevel}. In Section~\ref{sec:multilevel} we
show that it is collusion-resistant.

Consumer $i$ interacts optimally with the published query result
$r_i$ to get a result tailored specifically for it. In
Section~\ref{sec:umm}, we prove that the interaction yields optimal
utility for the consumer. The main idea is that the optimal
mechanism can be factored into two parts -- The first is a database
specific mechanism which has access to the database but not to the
user parameters. In our case this is the $\alpha_i$-geometric
mechanism. The second is the user specific mechanism, which has
access to the user loss-function and side-information and the
perturbed query result (given by the first mechanism), but not to
the database itself. Table~\ref{tab:optmech} shows these two
factors of the optimal mechanism discussed in
Section~\ref{sec:opt}.

We briefly discuss proof techniques: Section~\ref{sec:char}
completely characterizes mechanisms derivable from the geometric
mechanism using linear algebraic
techniques. Section~\ref{sec:applications} applies this
characterization twice: the first application shows that a
$\alpha$-geometric mechanism can be derived by re-randomizing the
output of a $\beta$-geometric mechanism so long as $\alpha >
\beta$. The second application shows that the mechanism induced by
the interaction of a rational information consumer with the
geometric mechanism is an optimal solution to the linear program
mentioned in Section~\ref{sec:opt}.

\subsection{Comparison with Bayesian Information Consumers}

An alternative to the Minimax decision rule is the Bayesian
decision rule. Ghosh et al.~\cite{Bayesian} prove an analogous
result to Theorem~\ref{thm:umm} for all Bayesian information
consumers. We briefly compare the models and the proof techniques.

The main distinction between the two models is their treatment of
side-information. The Bayesian model requires agents to have a
prior over all possible scenarios (true query results). Often, in
practice, agents do not behave consistent with the preferences of
the Bayesian model, perhaps because they find it hard to come up
with meaningful priors~\cite[Example 6.B.2,6.B.3]{BigFat}, or are
genuinely risk-averse~\cite[Section 6.3]{BigFat}.


As discussed in~\cite{Bayesian}, Bayesian information consumers
employ deterministic post-processing, unlike minimax information
consumers which employ randomized post-processing (For example, see
Table~\ref{tab:optmech}). Handling this extra complexity requires
us to construct a broader characterization of mechanisms derivable
from the geometric mechanism---Section~\ref{sec:char} presents a
complete characterization in terms of a simple condition on the
probability masses $x_{i-1,j},\: x_{i,j},\: x_{i+1,j}$. Our proof
avoids the LP based techniques and counting arguments
of~\cite{Bayesian}, and consequentially strictly generalizes and
gives a simpler proof of the main result of that paper. In
addition, our characterization enables us to release data at
multiple levels of privacy in a collusion-resistant manner.

\subsection{Related Work}
\label{sec:related}

A recent thorough survey of the state of the field of differential
privacy is given in~\cite{Survey08}. Dinur and
Nissim~\cite{Impossibility}, Dwork et al.~\cite{Decoding} establish
upper-bounds on the number of queries that can be answered with
reasonable accuracy. Most of the differential privacy literature
circumvents these impossibility results by focusing on interactive
models where a mechanism supplies answers to only a sub-linear (in
$n$) number of queries. Count queries
(e.g.~\cite{Impossibility,Datamining}) and more general queries
(e.g.~\cite{BigBang,Smooth}) have been studied from this
perspective.

Hardt and Talwar~\cite{Geometry} give tight upper and lower bounds
on the amount of noise needed to ensure differential privacy for
$d$ non-adaptive linear queries, where the database is a vector in
$\real^n$. Hay et al.~\cite{Boosting} give a way to increase
accuracy of answering multiple related queries while ensuring that
the query results follow consistency constraints. 

Blum et al.~\cite{Learning} focus attention to count queries that
lie in a restricted class; they obtain non-interactive mechanisms
that provide simultaneous good accuracy (in terms of worst-case
error) for all count queries from a class with polynomial VC
dimension. Kasiviswanathan et al.~\cite{Kavi08} give further
results for privately learning hypotheses from a given class.

The use of abstract ``utility functions'' in McSherry and
Talwar~\cite{MD} has a similar flavor to our use of loss-functions,
though the motivations and goals of their work and ours are
unrelated.  Motivated by pricing problems, McSherry and
Talwar~\cite{MD} design differentially private mechanisms for
queries that can have very different values on neighboring
databases (unlike count queries); they do not consider users with
side information and do not formulate a notion of mechanism
optimality (simultaneous or otherwise).

Our formulation of the multiple privacy levels is similar to Xiao
et al.~\cite{Multilevel}. However, they use random output
perturbations for preserving privacy, and do not give formal
guarantees about differential privacy.

\section{Characterizing Mechanisms\\  Derivable from the Geometric
  Mechanism}
\label{sec:char}

In this section we give a characterization of all mechanisms that
can be derived from the geometric mechanism. Recall that
differential privacy imposes conditions on every two consecutive
entries $(x_1, x_2)$ of every column: $x_1 \geq \alpha x_2$ (and
$x_2 \geq \alpha x_1$). Our characterization imposes syntactically
similar conditions on every three consecutive entries
$(x_1,x_2,x_3)$ in a column: $(x_2 - \alpha\cdot x_3) \geq \alpha
(x_1 - \alpha\cdot x_2)$. Neither condition implies the other. This
characterization is both necessary and sufficient for any
differentially private mechanism to be derivable from the geometric
mechanism. 

We defined feasible consumer interactions in
Section~\ref{sec:interaction}. A slightly different way of
representing these is to arrange the probability masses in a
$n\times n$ matrix $(T_{r,r'})_{r, r' \in N}$. We say that a matrix
is (row) stochastic if the sum of elements in each row is $1$ and
all elements are non-negative. We say that a matrix is a
generalized (row) stochastic matrix if the if the sum of elements
in each row is $1$, but with no condition on individual entries. If
the deployed mechanism is given by the matrix $y$, and the
reinterpretation by the matrix $T$, then the new mechanism is given
by the matrix $x = y \cdot T$.

We define a version of the Geometric Mechanism whose range is
restricted to $\{0,\ldots,n\}$, which will be easier to work with
since it can be easily represented as a matrix.

\begin{definition}[Range-Restricted Geometric Mechanism] 
  \label{def:truncgeom} 
  For a given privacy parameter $\alpha$, when the true query
  result is $k \in [0,n]$, the mechanism outputs $Z(k)$ where
  $Z(k)$ is a random variable with the following distribution for
  each integer $z$:
  \[
  Pr[Z(k) = z] =\begin{cases}
    \frac{1}{1+\alpha}\cdot \alpha^{|z-k|}  & \text{ if } z \in \{0,n \} \\
    \frac{1-\alpha}{1+\alpha} \cdot  \alpha^{|z-k|}  & \text{ if } 0 < z < n \\
    0 & \text{otherwise}.

  \end{cases}\]
  
\end{definition}

This mechanism is equivalent to the geometric mechanism in the
sense that we can derive this from the geometric mechanism and
derive the geometric mechanism from its range-restricted
version. We shall refer to both as the Geometric Mechanism and
denote the matrix by $G_{n,\alpha}$. (Table~\ref{tab:geom}).

For ease of notation, we shall denote by $G'_{n,\alpha}$ the matrix
obtained by multiplying the columns $1$ and $n$ of $G_{n,\alpha}$
by $(1+\alpha)$ and all other entries by
$\frac{1+\alpha}{1-\alpha}$. Table~\ref{tab:geom} shows the
matrices of $G_{n,\alpha}$ and $G'_{n,\alpha}$. We are now ready to
state the characterization.

\begin{table*}[ht]
  \begin{tabular}{cc}
    {
      $\frac{1-\alpha}{1+\alpha}
      \begin{bmatrix}
        \frac{1}{1-\alpha}\cdot 1        & \alpha & \alpha^2 &\ldots & \frac{1}{1-\alpha}\cdot\alpha^{n-1}\\ 
        \frac{1}{1-\alpha}\cdot\alpha   & 1      & \alpha   & \ldots & \frac{1}{1-\alpha}\cdot\alpha^{n-2}\\ 
        \frac{1}{1-\alpha}\cdot\alpha^2 & \alpha & 1        & \ldots \\
        \vdots & & & \ddots \\
        \frac{1}{1-\alpha}\cdot\alpha^{n-1} & \alpha^{n-2} & &\ldots &
        \frac{1}{1-\alpha}\cdot 1\\ 
      \end{bmatrix}$        
    }
    &
    {
      $
      \begin{bmatrix}
        1        & \alpha & \alpha^2 &\ldots & \alpha^{n-1}\\ 
        \alpha   & 1      & \alpha   & \ldots & \alpha^{n-2}\\ 
        \alpha^2 & \alpha & 1        & \ldots \\
        \vdots & & & \ddots \\
        \alpha^{n-1} & \alpha^{n-2} & &\ldots &   1\\ 
      \end{bmatrix}
      $
    } \\    
    $G_{n,\alpha}$ & $G'_{n,\alpha}$ \\
  \end{tabular}
  \caption{The Range Restricted Geometric Mechanism}
  \label{tab:geom}
\end{table*}

\begin{theorem}
  \label{thm:char}
  Suppose $M$ is any oblivious differentially private
  mechanism. Then $M$ can be derived from the geometric mechanism
  if and only if every three consecutive entries $x_1,x_2,x_3$ in
  any column of $M$ satisfy $ (x_2 - \alpha x_1) \geq \alpha(x_3 -
  \alpha x_2)$.
\end{theorem}

The key insight is to think of each column in $M$ and in
$G_{n,\alpha}$ as a vector. Looking at the problem through this
linear algebraic lens, we see that deriving $M$ from $G_{n,\alpha}$
amounts to proving that each column of $M$ lies in the convex hull
of the (vectors which form the) columns of $G_{n,\alpha}$. In
Lemma~\ref{lem:detg}, we show that $G_{n,\alpha}$ is non-singular,
hence each column of $M$ can be represented as a linear combination
of columns of $G_{n,\alpha}$.

\begin{restatable}{lemma}{nonsingular}
  \label{lem:detg}
  $\det(G_{n,\alpha}) > 0$. 
\end{restatable}
\begin{proof}
  Since $G'_{n,\alpha}$ can be obtained by multiplying each entry
  in the first and last column of $G_{n,\alpha}$ by $(1+\alpha)$
  and entries in all other columns by $\frac{1+\alpha}{1-\alpha}$,
  $\det{G'_{n,\alpha}} = (1+\alpha)^2
  (\frac{1+\alpha}{1-\alpha})^{n-2}\det{G_{n,\alpha}}$. Hence, we
  only need to prove that $\det{G'_{n,\alpha}} > 0$. We prove this
  by induction on $n$. For $n=2$, we explicit calculation yields
  $G'_{2,\alpha} = (1-\alpha^2)$.  For the general case, perform
  the column transformation $C_1 \leftarrow C_1 - \alpha C_2$ on
  $G'_{n,\alpha}$. Expanding on the first column gives us
  $\det{G'_{n,\alpha}} = (1-\alpha^2)\det{G'_{n-1,\alpha}}$. Hence,
  by induction, $\det{G'_{n,\alpha}} = (1-\alpha^2)^{n-1}$.
\end{proof}

We need to show that each column of $M$ is actually a convex
combination of columns of $G$. We can write $M = G_{n,\alpha} \cdot
T$ for some matrix $T$. Hence, $T = G_{n,\alpha}^{-1} \cdot
M$. Note that $G_{n,\alpha}$ and $M$ are both generalized
stochastic matrices. Since the set of all non-singular generalized
stochastic matrices forms a group~\cite{stochastic},
$G_{n,\alpha}^{-1}$ is a generalized stochastic matrix. And since
generalized stochastic matrices are closed under multiplication,
$T$ is also a generalized stochastic matrix and is uniquely
defined. All we need to prove is that all entries in $T$ are
non-negative. We shall use Cramer's Rule to calculate the entries
of $T$ and complete the proof.

Given a $n\times n$ matrix $G$ and a vector $x =
(x_1,\ldots,x_n)^t$, define $G(i,x)$ as the matrix where the
$i^{th}$ column of $G$ has been replaced by $x$.
  
Let $t_j$ be the $j^{th}$ column of $T$. $t_{i,j}$ denotes the
$i,j$ entry in $T$. Observe that, $G_{n,\alpha} \cdot t_j =
m_j$. By Cramer's Rule, we get that $t_{i,j} =
\frac{\det{G_{n,\alpha}(i,m_j)}}{ \det{G_{n,\alpha}}}$. To
calculate this, we shall explicitly calculate the value of
$\det{G_{n,\alpha}(i,m_j)}$.

\begin{restatable}{lemma}{determinantreplace}
  \label{lem:detreplace}
  Given $G_{n,\alpha}$ and a vector $x =
  (x_1,\ldots,x_n)^t$: 
  \begin{enumerate}
  \item $\det{G_{n,\alpha}(1,x)} > 0$ iff $x_1 > \alpha x_2$
  \item $\det{G_{n,\alpha}(n,x)} > 0$ iff $x_n > \alpha x_{n-1}$
  \item $\det{G_{n,\alpha}(i,x)} \geq 0$ if and only if $ (x_2 -
    \alpha x_1) \geq \alpha(x_3 - \alpha x_2)$ : For $ 2 \leq i
    \leq n-1$
  \end{enumerate}
\end{restatable}

Hence, when $M$ satisfies the condition that for every three
consecutive entries $x_1,x_2,x_3$ in any column $ (x_2 - \alpha
x_1) \geq \alpha(x_3 - \alpha x_2)$, then $s_{i,j} \geq 0$ for all
$i,j$.  This proves that $M$ can be derived from the geometric
mechanism.

To prove the converse, suppose that there is a column $c$ and row
$i$ of $M$ such that $((1+\alpha^2)m_{i,j} -\alpha(m_{i-1,j} +
m_{i+1,j})) < 0$, then $s_{i,c} = \det{G(i,m_c)} / \det{G} <
0$. This says that $M$ cannot be derived from $G$.  This completes
the proof of Theorem~\ref{thm:char}. {$\hfill \qed$\\}


We now prove Lemma~\ref{lem:detreplace}, using similar column
transformations as we used in Lemma~\ref{lem:detg} to calculate
$\det{G_{n,\alpha}(i,x)}$ for an arbitrary vector $x$.

\determinantreplace*
\begin{proof} We will prove the above properties for
  $G'_{n,\alpha}$. Since, $G'_{n,\alpha}$ is obtained from
  $G_{n,\alpha}$ by multiplying columns with positive reals, the
  properties above will continue to hold for $G_{n,\alpha}$. We
  divide the proof into cases depending on the value of $i$ :
\begin{enumerate}
\item{$i = 1$ : } We repeatedly do the column transformation $C_n
  \leftarrow C_n - \alpha C_{n-1}$ to get that
  $\det{G'_{n,\alpha}(1,x)} = (1-\alpha^2)^{n-2}
  \begin{vmatrix}x_{1} & \alpha \\ x_2 & 1\end{vmatrix}$ =
  $(1-\alpha^2)^{n-2}(x_1 - \alpha x_2)$. Hence,
  $\det{G'_{n,\alpha}(1,x)} > 0 \iff ( x_1 > \alpha x_2 )$.

\item{$i = n$ :} We can do the same column transformations to get
  that $\det{G'_{n,\alpha}(n,x)} = (1-\alpha^2)^{n-2}
  \begin{vmatrix}1 & x_{n-1} \\ \alpha & x_n\end{vmatrix}$ = $
  (1-\alpha^2)^{n-2}(x_n - \alpha x_{n-1})$. Hence,
  $\det{G'_{n,\alpha}(n,x)} > 0 \iff (x_n > \alpha x_{n-1})$.

\item{$2 \leq i \leq n-1$ : } Similarly, for the general case we
  get that $\det{G'_{n,\alpha}(i,x)} = (1-\alpha^2)^{n-3} 
  \begin{vmatrix}
    1 & x_{i-1} & \alpha^2 \\ 
    \alpha &x_i & \alpha \\
    \alpha^2 & x_{i+1} & 1
  \end{vmatrix} = (1-\alpha^2)^{n-2}((1+\alpha^2)x_i
  -\alpha(x_{i-1} + x_{i+1}))$. Hence, $\det{G'_{n,\alpha}(i,x)} >
  0 \iff (x_2 - \alpha x_1) \geq \alpha(x_3 - \alpha x_2)$.
\end{enumerate}
\end{proof}

\section{Applications of the Characterization}
\label{sec:applications}

We show two applications of the characterization result of
Theorem~\ref{thm:char}. The first one gives us a way to
simultaneously release data to consumers at different levels of
privacy. As a second application we show how to obtain a optimal
mechanism for an information consumer without knowing its
parameters.

\subsection{Information-Consumers at Different Privacy Levels}
\label{sec:multilevel}

Suppose we want to release the answer of the query to different
information consumers. We represent the level of privacy of a
consumer $c$ by the privacy parameter $\alpha_c$. Given true result
$r$, we will release $r_c$ to consumer $c$ such that the mechanism
is $\alpha_c$-differentially private. We expect that consumers at
different levels of privacy do not share query results with each
other which is enforced via, say, non-disclosure agreements. Even
when they do share data, we want our mechanism to be
collusion-resistant and not leak privacy-- the colluding group
should not get any more information about the database than the
consumer with access to the least private result i.e., the one with
the smallest $\alpha$ .

We now describe a mechanism that achieves this. The next lemma
gives us a way to ``add'' more privacy to an existing geometric
mechanism.

\begin{restatable}{lemma}{addprivacy}
  \label{lem:addpriv}
  For two privacy parameters $\alpha \leq \beta$, the geometric
  mechanism $G_{n,\beta}$ can be derived from the mechanism
  $G_{n,\alpha}$ i.e., there exists a stochastic matrix
  $T_{\alpha,\beta}$ such that $G_{n,\beta} = G_{n,\alpha} \cdot
  T_{\alpha,\beta}$.
\end{restatable}
\begin{proof}
  Theorem~\ref{thm:char} states that $G_{n,\beta}$ can be derived
  from $G_{n,\alpha}$ if and only if for every three consecutive
  entries $x_1,x_2,x_3$ in any column of $G_{n,\beta}$, $ (x_2 -
  \alpha x_1) \geq \alpha(x_3 - \alpha x_2)$. We check this
  condition for each of the three forms that consecutive entries in
  each row of $G_{\beta,n}$ can have:
  \begin{enumerate}
  \item $(\beta^i, \beta^{i+1}, \beta^{i+2})$ : $(1+\alpha^2)\beta^{i+1}
    - \alpha(\beta^i + \beta^{i+2}) = \beta^i (\beta + \alpha^2
    \beta - \alpha - \alpha \beta^2) = \beta^i (\beta - \alpha)
    (1-\alpha \beta) > 0. $
  \item $ (\beta, 1 , \beta) : (1+\alpha^2) 1 -  \alpha(\beta +
    \beta) = 1 + \alpha^2 - 2\alpha\beta > (1-\alpha)^2 > 0. $
  \item $(\beta^{i+2}, \beta^{i+1}, \beta^i) :
    (1+\alpha^2)\beta^{i+1} - \alpha(\beta^i + \beta^{i+2}) =
    \beta^i (\beta + \alpha^2 \beta - \alpha - \alpha \beta^2) =
    \beta^i (\beta - \alpha) (1-\alpha \beta) > 0. $
  \end{enumerate}
  This shows that $T_{\alpha,\beta} = G_{n,\alpha}^{-1} \cdot
  G_{n,\beta}$ is a stochastic matrix.
\end{proof}


\begin{algorithm}{ht}
  \dontprintsemicolon
  \KwIn{True Query Result $r$. $k$ privacy levels given by
    parameters $\alpha_1 < \alpha_2 < \ldots < \alpha_k$.}
  \KwOut{Query Results $r_1, r_2, \ldots, r_k$ to be released.}
  
  Define $T_1 = G_{\alpha_1,n}$.\; 
  \For{ $1\leq i \leq k$}
  { 
    Compute post-processing matrix $T_{i+1}$ such that
    $G_{\alpha_{i+1},n} = G_{\alpha_i,n} \cdot T_{i+1}$. 
  }\;

  By Lemma~\ref{lem:addpriv}, each $T_i$ is a stochastic
  matrix. Hence, we can think of $T_i$ as a mechanism -- Given any
  input $k$ we sample from the probability distribution given by
  the $k^{th}$ row of $T_i$ which we represent by $T_i(k)$ \;

  Let $r_0 = r$.\; 
  \For{ $1\leq i \leq k$}
  { $r_i = T_i(r_{i-1})$ is obtained by treating $r_{i-1}$ 
    as the true query output and applying mechanism $T_i$ to it.\;
  }
   
  Release the query results $r_1, r_2, \ldots, r_k$ to consumers at
  privacy levels $\alpha_1,\ldots,\alpha_k$.\;

  \caption{Releasing Query Result to Consumers at Multiple Levels
    of Trust.}
\label{alg:multilevel}
\end{algorithm}

The release mechanism is given in Algorithm~\ref{alg:multilevel}.
We conclude the section by proving that
Algorithm~\ref{alg:multilevel} is collusion-resistant.

\begin{lemma}
  \label{lem:multilevel}
  Any subset $C = \{c_1 < \cdots < c_t\} \subseteq \{1,\ldots,k\}$
  of colluding information consumers who have access to query
  results $R(C) = \{r_{c_1}\ldots r_{c_k}\}$, released at privacy
  levels $\alpha_{c_1}, \ldots, \alpha_{c_k}$, respectively, can
  only reconstruct as much information about the database $d$ by
  combining their results as $c_1$ can working alone.
\end{lemma}
\begin{proof}
  The matrix $G_{n,\alpha_1}$ and post-processing matrices
  $T_{\alpha_{c_i},\alpha_{c_{i+1}}}$ can be calculated by
  anyone. Hence, given the random coin tosses made by the
  algorithm, Lemma~\ref{lem:addpriv} shows that $r_{c_j}$ can be
  obtained from $r_{c_i}$ for $c_j > c_i$. Given $r_{c_1}$, having
  access to $R(C)$ can at most reveal information about these coin
  tosses that Algorithm~\ref{alg:multilevel} made. Since, these
  coin tosses do not depend on the database, any information about
  the database that is reconstructed from $R(C)$ can also be
  reconstructed by consumer $c_1$ (who has access to result
  $r_{c_1}$) alone.
\end{proof}

\subsection{Universal Utility Maximizing Mechanisms}
\label{sec:umm}

We now prove that if we deploy the geometric mechanism
(Definition~\ref{def:truncgeom}), then the interaction of every
information consumer will yield a mechanism that is optimal for
that consumer. Since, the geometric mechanism is not dependent on
any information consumer's loss-function or side information, it is
simultaneously optimal for all of them.

Our result proves that all optimal mechanisms can be derived from
the geometric mechanism. However, there do exist differentially
private mechanisms (which are not optimal for any information
consumer) that \emph{cannot} be derived from the geometric
mechanism. We give an example of such a mechanism in
Appendix~\ref{sec:nonmonotone}.

The first part of the proof shows that every two adjacent rows of
every optimal mechanism must satisfy certain condition; if it does
not, we can perturb the mechanism in a way to yield a
differentially private mechanism with strictly better utility. The
second part of the proof leverages this lemma and the
characterization from Theorem~\ref{thm:char} to complete the proof
of Theorem~\ref{thm:umm}.
\begin{lemma}
  \label{lem:perturb}
  For every monotone loss-function $L(l,S) = \max_{i\in S} \sum_{r
    \in N} l(i,r) \cdot x_{i,r}$, there exists an optimal mechanism
  $x$ such that for every two adjacent rows $i,i+1$ of this
  mechanism, there exist column indices $c_1$ and $c_2$ such that:
  \begin{enumerate}
  \item $\forall j \in 1...c_1$ : $\alpha x_{i,j} = x_{i+1,j}$
  \item $\forall j \in c_2...n$ : $x_{i,j} = \alpha x_{i+1,j}$
  \item Either $c_2=c_1+1$ or $c_2=c_1+2$. 
  \end{enumerate}
  
 \end{lemma}
\begin{proof}
  We define the function $L' : M \to \real$ given by $L'(x) =
  \sum_{i\in N} \sum_{r\in N} x_{i,r} \cdot |i-r|$. Consider the
  total order $\succ$ on $\real^2$ given by $(a,b) \succ (c,d)
  \iff\{ (a > c) \text{ or } (a = c \text{ and } b > d) \}$. Let
  $x$ be an optimal mechanism for the loss-function $(L,L')$
  according to the order defined above. The idea here is that there
  are multiple mechanisms that optimize $L$ and using $L'$ we
  isolate the ones with the property that we want. We prove by
  contradiction that $x$ satisfies the constraints given above.
  
  Assume otherwise. Then there exist rows $i,i+1$ and columns
  $j,k$; $k>j$ such that $\alpha x_{i,j} < x_{i+1,j}$ and $\alpha
  x_{i+1,k} < x_{i,k}$. We shall construct a differentially private
  mechanism $y$ for which $(L_y,L'_y)$ is strictly smaller than
  $(L_x,L'_x)$ which is a contradiction since we assumed that $x$
  minimized $(L,L')$.
  
  We divide the proof into two cases : $i \leq (j + k)/2$ and $i >
  (j + k)/2$. Consider the case $i \leq (j + k)/2$ first. For
  $i'\in \{1\ldots i\}$ set $y_{i',j} \leftarrow x_{i',j} + \delta
  x_{i',k}$ and $y_{i',k} \leftarrow (1-\delta)x_{i',k}$. For all
  other values set $y_{l,m} = x_{l,m}$. We first show that $y$ is a
  differentially private mechanism. Let the set of changed elements
  $C =\{ y_{l,m} : m \in \{j,k\} \text{ and } l \leq i\}$. The set
  of unchanged elements $U$ is all the remaining $y_{l,m}$. All
  privacy constraints involving elements only from $U$ are
  satisfied since they were satisfied in $M$. The privacy
  constraints involving only elements in $C$ continue to hold since
  they are the same linear combinations of corresponding elements
  from $M$. We only need to check that the privacy constraints are
  satisfied when one element is from $C$ and another from $U$. But
  this only happens for $y_{i,j}, y_{i+1,j}$ and $y_{i,k},
  y_{i+1,k}$. By assumption, $\alpha x_{i,j} < x_{i+1,j}$ and
  $\alpha x_{i+1,k} < x_{i,k}$. We can choose a small enough
  $\delta$ such that $\alpha y_{i,j} =\alpha (x_{i,j} + \delta
  x_{i,k}) < x_{i+1,j} = y_{i+1,j}$ and $ \alpha y_{i,k} = \alpha
  (1-\delta)x_{i,k} < x_{i+1,k} = y_{i+1,k}$. Also, for $m \in \{j,
  k\}$, $y_{i,m} > x_{i,m} > \alpha x_{i+1,m} = y_{i+1,m}$ . This
  proves that $y$ satisfies differential privacy. 

  Now, we shall prove that $y$ has strictly smaller loss than
  $x$. For any row $r \in \{1,\ldots,i\}$, the change in loss due
  to row $r$ is
  \begin{align*}
    &\sum_{i\in N} l(r,i) x_{r,i} - \sum_{i\in N} l(r,i) y_{r,i} \\
    &= ( l(r,j) x_{r,j} + l(r,k) x_{r,k}) - ( l(r,j) y_{r,j} + l(r,k)
    y_{r,k} ) \\
    &= ( l(r,j) x_{r,j} + l(r,k) x_{r,k}) \\
    &\qquad - ( l(r,j) (x_{r,j} +
    \delta x_{r,k}) + l(r,k) (x_{r,k} -\delta x_{r,k} ) \\
    &= \delta x_{r,k} ( l(r,k) - l(r,j)) \\
    &\geq 0 \quad \text{since
      $l(i,j)$ is monotonic in $|i-j|$}.
  \end{align*}
  The total loss $L = \max_{r\in S} \sum_{i\in N}l(i,r)\cdot
  x_{i,r}$ and since so $L_x \geq L_y$. Also $\sum_i \sum_r x_{i,r}
  \cdot |i-r| > \sum_i \sum_r y_{i,r} \cdot |i-r|$. This means that
  $(L_x,L'_x) \succ (L_y ,L'_y)$. But $x$ was an optimal mechanism
  with respect to $\succ$. This gives us a contradiction.
  
  The proof for the case $i > (j + k)/2$ is similar. For $i \geq i'$
  set $y_{i',k} \leftarrow x_{i',k} + \delta x_{i',j}$ and $y_{i',j}
  \leftarrow (1-\delta)x_{i',j}$. The same arguments as above now hold
  for this definition of $y$ as well.
\end{proof}

We are now ready to prove Theorem~\ref{thm:umm}. We state it again
for convenience.

\ummthm*

\begin{proof} 
  Algorithm~\ref{alg:multilevel} is used to deploy geometric
  mechanism at different levels of privacy. Lemma~\ref{lem:addpriv}
  shows that it is always possible to deploy geometric mechanism
  this way. This proves that the deployed mechanisms are
  differentially private. Lemma~\ref{lem:multilevel} proves that
  the release is $\alpha_{c_{i'}}$-differentially private even for
  any set $C$ of colluding consumers, where $i' = \min\{ j : j \in
  C\}$. This completes the proof of part 1.

  To prove part 2, we concentrate on a single trust level with privacy
  parameter $\alpha$. We prove the result by contradiction. Assume
  there is an information consumer $c$ with loss-function $l$ and
  side-information $S$, whose interaction with $G_{n,\alpha}$ does not
  optimize its loss. Let $M$ be an optimal differentially private
  mechanism for $c$ that satisfies Lemma~\ref{lem:perturb}. Since, $c$
  cannot optimize its loss by interacting with $G_{n,\alpha}$, $M$
  cannot be derived from the geometric mechanism. We prove that this
  implies that $M$ is infeasible which is a contradiction.

  We know from Theorem~\ref{thm:char} that
  there exists a column $j$ of $M$ and rows $i,i+1,i+2$, such
  that the three entries $x_{i,j},x_{i+1,j},x_{i+2,j}$ satisfy 

  \begin{equation} \label{eq:contra} (1+\alpha^2)x_{i+1,j} -
    \alpha(x_{i,j} + x_{i+2,j}) < 0.
  \end{equation}

  Recall the pattern of every pair of adjacent rows of $M$ from
  Lemma~\ref{lem:perturb}. Let $k$ be the unique column that
  satisfies $\alpha x_{i,k} < x_{i+1,k}$ and $\alpha x_{i+1,k} <
  x_{i,k}$, or if there is no such column, let it be the last
  column such that $\alpha x_{i,j} = x_{i+1,j}$. Let $a=\sum_{l <k}
  x_{i,l}$, $b = x_{i,k}$, $b' = x_{i+1,k}$, $b''= x_{i+2,k}$ and
  $c =\sum_{l >k} x_{i,l}$. Rewrite Equation~\eqref{eq:contra} to
  get: $ 0 \leq x_{i+1,j} - \alpha x_{i+2,j} < \alpha ( x_{i,j} -
  \alpha x_{i+1,j}) \implies x_{i,j} > \alpha x_{i+1,j} $ Thus, by
  Lemma~\ref{lem:perturb}, $k \geq j$. We now claim that:

  \begin{equation} 
    \label{eq:1} 
    (1+\alpha^2)b' - \alpha(b+b'') < 0
  \end{equation}

  This is true from Equation~\eqref{eq:contra} if $k=j$. Otherwise
  rewrite Equation~\eqref{eq:contra} to get $0 \leq x_{i+1,j} -
  \alpha x_{i,j} < \alpha ( x_{i+2,j} - \alpha x_{i+1,j}) \implies
  x_{i+2,j} > \alpha x_{i+1,j}$. Thus, by Lemma~\ref{lem:perturb},
  it must be that $\alpha \cdot b''=b$, Further, by privacy $b \geq
  \alpha b'$ and so, $b > \alpha^2 b'$.  This proves the claim.

  Because $M$ is a generalized stochastic matrix, $\sum_{l} x_{i,l}
  = \sum_{l} x_{i+1,l}=1$. Thus, $a+b+c=1$ and $\alpha\cdot
  a+b'+c/\alpha=1$.  Using these equations, we have:
  \begin{equation} 
    \label{eq:2}
    a=\frac{1-b-\alpha+b'\alpha}{1-\alpha^2} \quad \mbox{and} \quad
    c=\frac{\alpha-\alpha^2+ b\alpha^2-b'\alpha}{1-\alpha^2}
  \end{equation}
  We now prove that $M$ is not feasible.
  \begin{align*}
    &\sum_{l} x_{i+2,l} \geq \alpha^2 \cdot a + b'' + c/\alpha^2 \\
    &= \frac{\alpha^3 - b \alpha^3 - \alpha^4 + b'\alpha^4 +
      b''\alpha-b''\alpha^3 + 1- \alpha -b' + b\alpha}{\alpha
      (1-\alpha^2)}\\ 
    &= \frac{1-\alpha+\alpha^2}{\alpha} + \frac{(b+b'')\alpha - b'
      (1+ \alpha^2)}{\alpha} \\
    &> 1
  \end{align*}

  The first step is from Equation~\eqref{eq:contra} and
  Lemma~\ref{lem:perturb}, the second is by Equation~\eqref{eq:2},
  the third is by rearranging and the fourth holds because the
  first summand is always at least $1$ and the second is strictly
  positive by Equation~\eqref{eq:1}.
\end{proof}

\section{Conclusion}

We give a minimax model of utility for information consumers that
is prescribed by decision theory. We show that for any particular
count query, the geometric mechanism is simultaneously optimal for
all consumers, assuming that consumers interact rationally with the
output of the mechanism.  This is particularly useful in publishing
aggregate statistics, like the number of flu infections in a given
region, to a wide unknown audience, say on the Internet.

An open question is to investigate whether similar guarantees are
possible for multiple queries and other types of queries.

\bibliographystyle{plain}
\bibliography{priv}

\appendix

\section{There always exists an Optimal Mechanism that is
  Oblivious}
\label{sec:nonoblivious}

In Section~\ref{sec:model} we restricted attention to oblivious
mechanisms. While natural mechanisms (such as the Laplace mechanism
from~\cite{BigBang}) are usually oblivious, we now justify this
restriction from first principles. Specifically, we show that for
every information consumer with a loss-function over databases and
side information over query results, there exists a oblivious
loss-function and side-information, such that the optimal utility
with the oblivious loss-function is no more than the optimal
utility with the non-oblivious loss-function.

Consider a non-oblivious mechanism $x$. For the minimax information
consumer with loss-function $l$ over databases and side information
$S \subseteq \{0,1,\ldots,n\}$, the utility of this mechanism is
given by
\begin{equation} 
  \label{eq:obj1}
  \max_{d\in S \subseteq D^n} \sum_{r\in N} x_{d,r}  \cdot l(f(d),r)
\end{equation}

The following lemma proves that obliviousness is without loss of
generality i.e. there always exists an oblivious mechanism whose
loss is lower than or equal to the loss of the best non-oblivious
mechanism.

\begin{lemma} 
  Fix a database size $n \geq 1$ and privacy level $\alpha$. For
  every minimax information consumer with loss-function $l$ and
  side information $S \subseteq \{0,1,\ldots,n\}$, there is an
  $\alpha$-differentially private mechanism that minimizes the
  objective function~\eqref{eq:obj1} and is also oblivious.
\end{lemma}

\begin{proof}
  We shall now construct a differentially privacy
  mechanism $x'$ that is oblivious and whose loss is not greater
  than the loss of $x$. This will prove our assertion.
 
  We construct a partition $E$ of all the databases, according to
  the query output. All databases that have the same query output
  belong to the same subset of the partition. For a database $d$,
  let $E(d) = \{d' : q(d) = q(d')\}$. For $r\in N$ and $d\in D^n$,
  define $x'_{E(d),r} = \avg_{d' \in E(d)} x_{d',r}$. It is clear
  that $x'$ is an oblivious mechanism.

  First we show that $x'$ is $\alpha$-differentially private. Fix
  two databases $d_1, d_2 \in D^n$ such that $d_1$ and $d_2$ differ
  in exactly one row; We need to show that $\alpha x'_{d_1r} \leq
  x'_{d_2r}$. Assume $f(d_1) \ne f(d_2)$, otherwise the proof is
  trivial.

  For any database of $E(d_1)$, we can generate all its neighbors
  (databases that differ in exactly one row) in $E(d_2)$ by
  enumerating all the ways in which we can change the query result
  by exactly $1$. For instance when $f(d_1) = f(d_2) +1$, pick one
  of the $n-f(d_1)$ rows that satisfy the predicate in $d_1$ and
  change its value to one of those that violates the
  predicate. This process is identical for all databases of
  $E(d_1)$, and so for all $d \in E(d_1)$, the number of neighbors
  of $d$ that belong to the set $E(d_2)$ is the same (does not vary
  with $d$). Similarly, for all $d \in E(d_2)$, the number of
  neighbors of $d$ that belong to the set $E(d_1)$ is the same.

  Consider the following set of inequalities that hold because $x$
  is $\alpha$-differentially private: $d \in E(d_1)$, $d' \in
  E(d_2)$, where $d_1$ and $d_2$ are neighbors, $ \alpha x_{dr}
  \leq x_{d'r}$. By the argument in the above paragraph, all the
  databases in $E(d_1)$ appears equally frequently in the
  left-hand-side of the above inequality and all the databases in
  $E(d_2)$ equally frequently in the right-hand-side. Summing the
  inequalities and recalling the definition of $x'$ completes the
  proof of privacy.

  Now we show that $x'$ does not incur more loss than $x$. The loss
  for $x'$ is given by $\max_{d\in S \subseteq D^n} \sum_{r\in N}
  x'_{E(d),r} \cdot l(f(d),r)$. Suppose the worst loss for $x'$
  occurs for the partition $E(d_1)$.
  \begin{align*}
    L(x') &= \sum_{r \in N} x'_{E(d_1,r)} \cdot  l(f(d_1),r)\\
    &= \sum_{r \in N} \avg_{d \in E(d_1)}(x_{d,r}) \cdot l(f(d),r) \\
    &\leq \max_{d \in E(d_1)} \sum_{r\in N} x_{d,r} \cdot l(f(d),r) \\
    &\leq \max_{d \in S \subseteq D^n} \sum_{r\in N} x_{d,r} \cdot l(f(d),r) = L(x).
  \end{align*}
  This completes the proof.

\end{proof}






\section{A Differentially Private Mechanism that is not derivable
  from the Geometric Mechanism}
\label{sec:nonmonotone}

Consider the mechanism $M$ given by the following matrix. $M(i,j)$
gives the probability of returning $j$ when the true query result
is $i$. We can verify that $M$ is $\frac{1}{2}$-differentially private. 
\[ M = \begin{bmatrix}
        1/9   &     2/9    &    4/9  &      2/9 \\
        2/9   &     1/9    &    2/9  &      4/9 \\
        4/9   &     2/9    &    1/9  &      2/9 \\
      13/18   &     1/9    &   1/18  &      1/9 \\
\end{bmatrix}
\]
We claim that $M$ cannot be derived from the geometric
mechanism. We can explicitly calculate $G^{-1}_{3,\frac{1}{2}}\cdot
M$ to see that $M$ is not derivable from the geometric. Instead we
shall use the characterization from Theorem~\ref{thm:char}. If we
look at elements $M(0,1), M(1,1), M(2,1)$ , then
$(1+\alpha^2)M(1,1) - \alpha(M(0,1) + M(2,1)) = 1.25\times
\frac{1}{9} - \frac{1}{2} \times (\frac{2}{9} + \frac{2}{9}) =
\frac{- 0.75}{9}$. This proves that $M$ cannot be derived from
$G_{3,\frac{1}{2}}$.

\end{document}